\documentclass[reqno]{amsart}

\usepackage{amsmath}		
\usepackage{amssymb}		
\usepackage{amsfonts}		
\usepackage{amsthm}		

\usepackage{mathtools}		

\usepackage[foot]{amsaddr}		

\usepackage[utf8]{inputenc}		
\usepackage[T1]{fontenc}		

\usepackage[proportional,tabular,lining,sf,mono=false]{libertine}

\usepackage{acronym}		

\usepackage[svgnames]{xcolor}	
\colorlet{MyRed}{DarkRed!50!Crimson}
\colorlet{MyBlue}{DodgerBlue!75!black}
\colorlet{MyGreen}{DarkGreen}
\colorlet{MyViolet}{DarkMagenta}

\colorlet{MyLightBlue}{DodgerBlue!20}
\colorlet{MyLightGreen}{MyGreen!20}

\colorlet{PrimalColor}{MyBlue}
\colorlet{PrimalFill}{MyLightBlue}
\colorlet{DualColor}{MyRed}

\colorlet{AlertColor}{MyRed}	
\colorlet{BadColor}{MyRed}	
\colorlet{GoodColor}{MyGreen}	
\colorlet{LinkColor}{MediumBlue}	
\colorlet{MacroColor}{MyRed}
\colorlet{RevColor}{MediumBlue}	

\usepackage{twoxtwogame}

\usepackage{array}		
\usepackage{booktabs}		
\usepackage[inline,shortlabels]{enumitem}		
\setenumerate{itemsep=\smallskipamount,topsep=2pt,left=\parindent}
\setitemize{itemsep=\smallskipamount,topsep=2pt,left=\parindent}

\usepackage{comment}

\usepackage[round,authoryear]{natbib}		
\setcitestyle{aysep={},yysep={,}}

\usepackage{hyperref}
\hypersetup{
final,
colorlinks=true,
linktocpage=true,
pdfstartview=FitH,
breaklinks=true,
pdfpagemode=UseNone,
pageanchor=true,
pdfpagemode=UseOutlines,
plainpages=false,
bookmarksnumbered,
bookmarksopen=false,
bookmarksopenlevel=1,
hypertexnames=false,
pdfhighlight=/O,
urlcolor=LinkColor,linkcolor=LinkColor,citecolor=LinkColor,	
pdftitle={},
pdfauthor={},
pdfsubject={},
pdfkeywords={},
pdfcreator={pdfLaTeX},
pdfproducer={LaTeX with hyperref}
}

\newcommand{\EMAIL}[1]{\email{\href{mailto:#1}{#1}}}

\usepackage[sort&compress,capitalize,nameinlink]{cleveref}		
\crefname{algo}{Algorithm}{Algorithms}
\crefname{assumption}{Assumption}{Assumptions}
\crefname{figure}{Fig.}{Figs.}
\crefname{model}{Model}{Models}

\usepackage{algorithm}		
\usepackage{algpseudocode}		

\usepackage{thmtools}		
\usepackage{thm-restate}		

\theoremstyle{plain}
\newtheorem{theorem}{Theorem}		
\newtheorem{corollary}{Corollary}		
\newtheorem{proposition}{Proposition}		

\theoremstyle{definition}
\newtheorem{definition}{Definition}		

\newtheorem*{example*}{Example}		

\theoremstyle{remark}
\newtheorem{remark}{Remark}
\newtheorem*{remark*}{Remark}

\usepackage[showdeletions]{color-edits}		
\newcommand{\draft}[1]{#1}		

\newcommand{\bb}{\color{Black}}
\newcommand{\bc}{\color{Black}}

\newcommand{\newmacro}[2]{\newcommand{#1}{\draft{#2}}}		
\newcommand{\newop}[2]{\DeclareMathOperator{#1}{\draft{#2}}}		

\DeclarePairedDelimiter{\pospart}{[}{]_{+}}		

\DeclarePairedDelimiterX{\setof}[1]{\{}{\}}{#1}		
\DeclarePairedDelimiterX{\setdef}[2]{\{}{\}}{#1:#2}		
\DeclarePairedDelimiterXPP{\exclude}[1]{\mathopen{}\setminus}{\{}{\}}{}{#1}

\newop{\simplex}{\Delta}		

\newcommand{\alt}[1]{#1'}		
\newcommand{\altalt}[1]{#1''}		

\newmacro{\dd}{\:d}		

\newcommand{\bv}{\color{Black}}

\newmacro{\tstart}{0}		
\renewcommand{\time}{\draft{t}}		
\newmacro{\timealt}{s}		

\newop{\Nash}{NE}		
\newop{\brep}{br}		
\newop{\reg}{Reg}		
\newop{\val}{val}		

\newmacro{\players}{\mathcal{I}}		

\newmacro{\pure}{i}		
\newmacro{\purealt}{j}		
\newmacro{\purealtalt}{k}		
\newmacro{\nPures}{m}		
\newmacro{\pures}{S}		

\newmacro{\strat}{\mathbf{x}}		
\newmacro{\stratx}{x}
\newmacro{\stratcl}{c_l}	
\newmacro{\stratcr}{c_r}	
\newmacro{\stratalt}{\alt\strat}		
\newmacro{\strataltalt}{\altalt\strat}		
\newmacro{\strats}{\mathcal{X}}		
\newmacro{\intstrats}{\strats^{\circle}}		

\newmacro{\eq}{p}		

\newmacro{\pay}{u}		
\newmacro{\payv}{v}		
\newmacro{\pot}{\Phi}		

\newmacro{\game}{\mathcal{G}}		
\newmacro{\gamefull}{\game(\pures,\payv)}		

\newmacro{\fingame}{\Gamma}		
\newmacro{\fingamefull}{\Gamma(\players,\pures,\pay)}		
\newmacro{\mixgame}{\simplex(\fingame)}

\newmacro{\plans}{Y} 
\newmacro{\attr}{A}		

\newmacro{\plan}{y}		
\newmacro{\wattr}{\lambda_{\attr}}		
\newmacro{\relattr}{\Lambda}		
\newmacro{\awf}{\lambda}		
\newmacro{\dive}{v}		
\newmacro{\distinct}{d}		
\newmacro{\consumers}{\mathcal{I}}		
\newmacro{\consum}{z}	

\newmacro{\improf}{V}	
\newmacro{\cost}{C_{\inno}}	
\newmacro{\aprof}{\alpha}
\newmacro{\bprof}{\beta}
\newmacro{\inno}{I} 
\newmacro{\imit}{R} 
\newmacro{\reta}{RT} 
\newmacro{\marketv}{g} 
\newmacro{\marketvu}{\bar g} 

\newmacro{\domain}{D}
\newmacro{\mat}{M}		
\newmacro{\hmat}{H}		

\newmacro{\state}{\strat}		
\newmacro{\statealt}{\score}		

\renewcommand{\phi}{\varphi}
\newmacro{\fbio}{\bio}
\newmacro{\fecon}{\econ}
\newmacro{\fcs}{\learn}

\newop{\RD}{RD}
\newmacro{\set}{\mathcal{S}}
\newmacro{\size}{z}
\newmacro{\score}{y}
\newmacro{\step}{\delta}
\newmacro{\switch}{\rho}
\newmacro{\gain}{\gamma}
\newmacro{\payvec}{\pi}

\usepackage{multicol}
\usepackage{multirow,array}
\usepackage{twoxtwogame}

\newacro{LHS}{left-hand side}
\newacro{RHS}{right-hand side}
\newacro{iid}[i.i.d.]{independent and identically distributed}
\newacro{lsc}[l.s.c.]{lower semi-continuous}
\newacro{whp}[w.h.p]{with high probability}
\newacro{wp1}[w.p.$1$]{with probability $1$}
\newacro{ODE}{ordinary differential equation}

\newacro{CCE}{coarse correlated equilibrium}
\newacroplural{CCE}[CCE]{coarse correlated equilibria}
\newacro{NE}{Nash equilibrium}
\newacroplural{NE}[NE]{Nash equilibria}
\newacro{ESS}{evolutionarily stable state}

\newmacro{\gaina}{a}
\newmacro{\gainb}{b}
\newmacro{\gainc}{c}
\newmacro{\gaind}{d}
\newmacro{\PPI}{pairwise proportional imitation}

\begin{document}

\title{Diversity in Schumpeterian games}

\author[F.~Falniowski]{Fryderyk Falniowski$^{\ast}$}
\author[E.~Pli\'{s}]{El\.{z}bieta Pli\'{s}$^{\ast}$}
\EMAIL{falniowf@uek.krakow.pl}
\EMAIL{plise@uek.krakow.pl}

 \address{$^{\ast}$\, Department of Mathematics, Krakow University of Economics, Rakowicka 27, 31-510 Krak\'{o}w, Poland.}

\begin{abstract}
 We examine the impact of the change in diversity introduced by a~new product on the evolution of the economic system. Modeling Schumpeterian competition as a population game with the unique, attracting, evolutionarily stable state (the Schumpeterian state), in which both innovators and imitators coexist, we examine how the Schumpeterian state evolves depending on the properties of a new product developed by innovators. This way, we demonstrate that diversity change is one of the spiritus movens of innovation influencing the process of creative destruction.

\textbf{Keywords:} diversity, Schumpeterian competition, evolutionarily stable state

\textbf{JEL Classification Numbers:} C72, C73, D21, O31, O33
\end{abstract}

\maketitle

%

\section{Introduction}
\label{intro}
Since \citet{Schumpeter2017TheoryCycle,Schumpeter2017CapitalismDemocracy}, innovation, together with creative destruction, has been considered the driving force behind the evolution of the economic system. Thus, various models of economic evolution usually include concepts involving some form of innovation -- ranging from the Nelson-Winter model \citep{nelson1977search} and its extensions, with a vast literature in evolutionary economics \citep{Malerba1992LearningChange,Silverberg1988InnovationModel,Saviotti2004EconomicSectors}, to the role of innovation studied from the perspective of mechanism design and game theory \citep{lipieta2022diversity,belleflamme2011incentives,
pepall2014industrial}, and the Schumpeterian growth models  \citep{Acemoglu2008IntroductionGrowth, Matsuyama2022DestabilizingInnovation,Aghion1992ADestruction}. 
This process, in which old technologies are replaced by new ones and the creation of new solutions involves the elimination of existing ones, can lead, as Aghion and Howitt noted \citep{Aghion2009,Aghion1992ADestruction}, to sustainable economic growth.

But when does creative destruction make way for both {\it the old} and {\it the new}? Intuitively, in such a case, innovation must be relevant enough to be chosen by firms; on the other hand, using the old strategy must still have its benefits. We argue that the reason both options can exist in the market is the relationship between the {\it diversity} of the products: one introduced by an innovation, and the other already in use.
Thus, we postulate the interconnection between the diversity of the product market and the relevance of innovation.

Recently, research on the impact of diversity on the behavior of economic entities has attracted considerable interest (see, for example \citep{arbatli2020diversity,acemoglu2023learning}, but  
studying the nature of diversity goes far beyond economic thinking, as in the seminal works of  \citet{weitzman1992diversity,weitzman1998noah} and  \citet{Nehring2002ADiversity, Nehring2004ModellingProduction}. There is a wide variety of perspectives on diversity ranking, ranging from concepts of (bio)diversity in biology (for an overview of biodiversity measures, we refer the reader to \citep{magurran2013measuring}), through information-theoretic measures, with a special role for entropy-based diversity and complexity indexes (for diversity indexes based on entropy see \citep{jost2006entropy,falniowski2020entropy, Stirling2007ASociety}), to economically oriented concepts of diversity as a metric of opportunity (such perspective was studied e.g. by \citet{barbera2004ranking}), and product variety, along with its implications. (There is a vast literature considering the concept of product variety, for different perspectives we refer the reader to \citep{Stirling2007ASociety,stiglitz2017monopolistic,anderson1999pricing}).
Since holding on to the Schumpeterian vision of economic development \citep{Schumpeter2017CapitalismDemocracy}, innovation is the only way to create economic novelty; there can be no growth in economic diversity without innovation. Of particular interest, due to its natural ability to increase product diversity, is an innovation that offers consumers a previously unknown good. It can have features that 
turn out to be essential.
We call it a \emph{relevant innovation} (we take it after \citep{lipieta2022diversity}). 

While studying the mutual relationship between economic evolution, driven by innovation,  and economic diversity, one can ask \bv 1) \bc how a relevant innovation changes economic diversity and \bv 2) \bc how the change in diversity affects economic evolution. Regarding the first question, although, as already mentioned, the relevant innovation is a necessary condition for an increase in economic diversity, it is not sufficient (for a more detailed discussion, see \citep{plis2020}). Since innovation usually entails some kind of destruction (as new products are created, old technologies are no longer used), one may be concerned about a possible decline in diversity. However, if an innovative technology adopts all the relevant attributes of a previously used technology, the process of creative destruction does not necessarily lead to a reduction in economic diversity (we refer the reader to \citep{lipieta2022diversity} for the example of an eco-innovation that does not reduce economic diversity). 
This allows us to point to the diversity of a product as another component of the Schumpeterian paradigm, which invariably influences the modern economy \citep{aghion2019innovation,aghion2022creative,aghion2023creative}.

{\bf Our contribution.} In this paper, we 
 show how changes in diversity  of products influence the evolution of the economic system. In this way, we address the second question and examine how an increase in diversity due to relevant innovation can affect the ratio of innovators, thereby changing the structure of the economy. We propose the concept of economic diversity, in which consumer needs are revealed through the structure and weight of the products' attributes, assuming that consumers can evaluate the features of any good with respect to their internal hierarchy of needs (in the spirit of \citet{MengerCarl1871PrinciplesEconomics}). This approach enables us to identify and weigh all the relevant attributes that a product or a technology possesses and to  consider its diversity (according to the definition of a function of diversity proposed by \cite{Nehring2002ADiversity}). 
 In this way, it enables us to look at (relevant) innovation through the lenses of diversity (determined by relevant attributes) and dissimilarity of new products/technologies compared to existing ones. Then
the increase in diversity caused by innovation may encourage firms to innovate.  Thus, the change in diversity impacts payoffs of agents, and it is identified as an important factor of the evolution of the system.

More precisely, we consider a model of Schumpeterian competition that takes the form of a population game among firms in the given product market. 
We define two {\it Schumpeterian strategies} (using the notion of a Schumpeterian strategy from \citep{Andersen2007BridgingTheory}), which are available to any firm: \emph{to innovate} and \emph{to refrain from innovation}. While introducing a relevant innovation changes a product or technology in a way that consumers need \bv (and await)\bc, refraining from innovation means not making any significant changes to the good produced and the technology used. In the model, the payoff for a firm that uses one of the Schumpeterian strategies depends on how consumers evaluate the product modification compared to the old technology.
 We show that innovation that meets consumer needs by increasing economic diversity (i.e., the relevant innovation) can create an evolutionarily stable structure within the population of firms.
Moreover, such an evolutionary stable state of the population, in which some firms decide to innovate while others stick to proven technological solutions, is unique in each Schumpeterian game. The unique, stable structure thus created is called the Schumpeterian state. 

The increase in diversity associated with the relevant innovation affects the rate of firms that decide to make an effort to innovate in the Schumpeterian state.
The greater the increase in diversity that innovative technology provides, compared to the diversity of the technology previously used, the more firms decide to develop new technologies. Therefore, the more innovations are needed, the more firms will make efforts to introduce them. Thus, the change in diversity directly affects the evolution of the economic system. This indicates the need to take into account economic diversity in growth models based on the concept of innovation.

{\bf Related work.} 
There is a vast plethora of literature on diversity and evolution. As within the scope of this paper there is no review of the topic, we address the articles that particularly influenced our research or are of primary importance in the field. The notion of diversity that we apply was proposed by  \cite{Nehring2002ADiversity} and used in various contexts, from modeling product complementarities \citep{Nehring2004ModellingProduction} to evaluating agents' freedom of choice \citep{Sher2018EvaluatingFreedom} and measuring the technological diversity of the economic system \citep{plis2020}. Here, we extend the concept by constructing a function of diversity in which the weights of attributes are set by the (continuum) population of agents, thereby enabling us to apply it in a population game with a continuum of firms and a continuum of consumers.
The context and initial idea of the Schumpeterian game we propose are based on the unpublished work by \cite{Andersen2007BridgingTheory}. 

Innovation, and process innovation in particular, seen as a tool to reduce the average production cost of firms, has been widely studied, especially in the context of a duopoly \citep{delbono1991incentives,belleflamme2011incentives,ferreira2014incentives,parenti2017cournot}.
In particular, we refer to the work by \cite{RePEc:eee:gamebe:v:31:y:2000:i:1:p:1-25}, who
showed that an industry with a priori identical firms that engage in some level of R\&D ends up with a heterogeneous population with an innovator-imitator configuration at equilibrium that resembles our  Schumpeterian state.
On the other hand, the model of Schumpeterian competition we propose fits into broadly understood monopolistic competition with vertical product differentiation \citep{dixitst77,chang2012monopolistic} and invokes Lancaster's legacy of a characteristics-based demand approach (see: \citep{LancasterInnovative, Lancaster1966ATheory,magnolfi}.

\section{Diversity}
\label{div}
We introduce the concept of diversity in the game-theoretic setting. 
The approach presented there, based on the attribute model proposed by \citet{Nehring2002ADiversity}, allows for the measure of diversity to be built from the bottom up -- by leveling up from the idea of objects' attributes.

\subsection{The multi-attribute model of diversity}  Throughout the paper, we identify a good with the technology of its production and \bb denote a finite and non-empty \bc set of technologies that firms operating in the market can use, \bb or the products they offer, by $\plans$. \bc
We aim to define a function of diversity similar to \bb that proposed by \citet{Nehring2002ADiversity}, which, for the continuum of consumers,  would reflect their \bc needs and tastes for different technologies (and goods).  In this paper we show how product diversity and technology dissimilarities affect Schumpeterian competition. To this aim, in this section we introduce the building blocks of the model step by step.

\emph{An attribute} 
 in the \bv (nonempty) \bc set $\plans$ is any subset $\attr\subset \plans$ of technologies (or goods) that possesses defined characteristics (such as the material used, the color of the good produced, its weight, usage, etc.). \bb The attribute $\attr$ is then identified with the set \bc \[\{\plan\in \plans\colon  \plan \textrm{ possess attribute } \attr\}.\] 
Each element in the set $\attr$ \emph{possesses} the attribute $\attr$, that is, $\plan \in \attr$, if the technology
$\plan\in \plans$ possesses the attribute $\attr$.  \bb We say that \bc a set of technologies $D\subset \plans$ \emph{realizes} the attribute $\attr$, if there is at least one technology of $D$ that possesses the attribute $\attr$, that is, if $\attr\cap D \neq \emptyset$. \bb
The number of possible attributes available to consumers is then no greater than $2^{\kappa}-1$, where $\kappa$ is the number of technologies used by the firms. \bc

\bb In the context of a population game, we are dealing with an uncountable set of consumers who assign their own values to each attribute.
Let a unit interval $\players = [0, 1]$ be the set of consumers, each consumer $\consum\in\players$ assigning  \bc a value of $\awf_{\attr}(\consum)\in [0, 1]$ to the attribute $\attr\subset \plans$.
\bb The rate of consumers who assigned $\attr$ values greater than $k\in [0,1]$ is given by \bc the (Lebesgue) measure of the set $\players_k=\{\consum\in  \players\colon \awf_{\attr}({\consum})>k\}.$ 
\bb We refer to the number \bc\begin{equation}\label{02} \wattr:=\int_\consumers \awf_{\attr}({\consum})dz, 
\end{equation} \bb i.e., the total value assigned to  $\attr$ by consumers, as \bc\emph{the weight of the attribute}~$\attr$.
\bb Notice that the weight of $\attr$ \bc is equal to the mean \bb value \bc of the function $\consumers\ni\consum\mapsto \awf_{\attr}({\consum})\in[0,1]$ \bb on $\players$.  (One can see that in the case of $\players$ being a finite set, the formula \eqref{02} takes the form of $\wattr=\sum_{\consum\in \consumers}\awf_{\attr}(\consum)$, making the weight of the attribute $\attr$ the same as in \cite{Nehring2002ADiversity}.) 
 We call $\attr\subset\plans$ \emph{the relevant attribute} if the weight of $\attr$ is positive. The set of technologies (or goods) that consumers find relevant is referred to as \emph{the family of relevant attributes} and is denoted by \bc \begin{equation}
\label{rodzina1}\relattr:=\{ \attr\subset  \plans: \wattr > 0\}.  
\end{equation} 
 The function $\awf\colon 2^\plans\ni \attr\mapsto \wattr\in [0,1]$ given by \eqref{02} is called \emph{the attribute weighting function}. It determines the social value of any attribute $\attr \subset \plans$ and, according to the following definition proposed by \citet{Nehring2002ADiversity}, allows us to measure diversity of any product.

 \bc

\begin{definition}
A function $\dive\colon 2^{\plans}\rightarrow [0,+\infty)$, such that for any non-empty set $D\subset \plans$ 
\begin{equation}\label{d}
\dive (D)=\sum_{ \attr\subset  \plans: \attr\cap D\neq\emptyset} \wattr,\end{equation}
with $\dive (\emptyset):=0$, is called {\it the function of diversity}. The value of $\dive (D)$ is referred to as \emph{the diversity of the set} $D\subset \plans$. 
\end{definition}

\bb
Having the function of diversity $\dive\colon 2^{\plans}\rightarrow [0,+\infty)$ defined by consumers in the given product market, we can compare any two technologies $\plan,\plan'\in \plans$ in terms of diversity in the following way. \bc
Denote by \begin{equation}\label{dis}
\distinct(\plan',\plan):=\dive(\{\plan',\plan\})-\dive(\{\plan'\})=\sum_{\attr:\,  \plan' \in \attr,\,\plan\notin \attr} \wattr,   
\end{equation} 
i.e.,  the total weight of attributes possessed by the technology $\plan'$, that the technology $\plan$ does not have. We call it \emph{the dissimilarity of the technology $\plan'$ from the technology $\plan$} \citep{Nehring2002ADiversity}.  On the other hand, the dissimilarity of the technology $\plan$ from $\plan'$ is given by \begin{equation}\label{dis'}
\distinct(\plan,\plan')=\dive(\{\plan',\plan\})-\dive(\{\plan\})=\sum_{\attr:\,  \plan \in \attr,\,\plan'\notin \attr} \wattr,  
\end{equation}
 and, since different technologies may possess different attributes and different values of diversity,  usually, for $\plan\neq \plan'$, there is $\distinct(\plan,\plan')\neq \distinct(\plan',\plan)$. 

 The dissimilarities between two technologies $\plan$ and $\plan'$ given by \eqref{dis} and \eqref{dis'} will be decisive factors of the firm's payoffs when modeling Schumpeterian competition. We will show that the firm's decision whether to introduce innovation $\plan'$ or to stay with the old technology $\plan$ depends on the difference in their diversity. However, the increase in diversity connected with innovation can be measured by the products' dissimilarities, since \begin{equation}
    \label{przyrost} \dive(\{\plan'\})-\dive(\{\plan\})=\distinct(\plan',\plan)-\distinct(\plan,\plan').
\end{equation} 
The dissimilarities between "the old" and "the new" will influence the firms' behavior in Schumpeterian competition (see: Proposition \ref{p:rate} and Corollaries \ref{prop:s1} and \ref{cor:imit}). To show how product diversity and technology dissimilarities affect agents' payoffs in the  Schumpeterian (population) game, in the following section we describe the crucial concepts from evolutionary game theory.

\section{Population games}
We will study a population game  with a continuum of agents, who consider choosing a new production plan (an innovation) or imitating a known technology. Thus, in this section we introduce basic concepts along with the canonical textbook of \citet{San10}.   

Throughout the paper, we consider games with a continuum of nonatomic players, modeled by the unit interval $\players = [0, 1]$, with each player choosing (in a measurable way) an \emph{action} from the set of two available strategies (actions) $ \pures \equiv \setof{s_1,s_2}$.
Denoting by $\stratx_1 \in [0,1]$ the mass of agents playing the first strategy (thus $\stratx_2=1-\stratx_1$ are those who choose the second strategy), the overall distribution of actions
is specified by \emph{state of the population} $\strat = (\stratx_1,\stratx_2)$, being a point in the unit simplex $\strats$.
We are interested in the case of two individuals, randomly selected from the population, and matched to play a symmetric two-player game with a payoff matrix $\mat = (\mat_{\pure\purealt})_{\pure,\purealt\in\pures}$.
In this case, the payoff for agents playing $\pure\in\pures$ in state $\strat$ is given by $\pay_{\pure}(\strat) = \sum_{\purealt\in\pures} \mat_{\pure\purealt} \stratx_{\purealt}$.
Following standard conventions in the field (see e.g. \citep{San10,HLMS22}), we refer to this scenario as \emph{symmetric random matching}.
We denote by $\payv(\strat) = (\pay_{1}(\strat),\pay_{2}(\strat))$ the associated \emph{payoff vector} in state $\strat$,
and
we refer to the tuple $\game \equiv \gamefull$ as \emph{population game}.

 A state $\strat\in\strats$ is \emph{a Nash equilibrium} of the game $\game$ if every strategy used in $\strat$ earns a maximum payoff (equivalently, each agent in the population chooses an optimal strategy with respect to the choices of others).
 
 The state $\strat\in\strats$ is \emph{evolutionarily stable} (is \emph{ESS}) if (i) $\strat^T\mat\strat\geq \stratalt^T\mat\strat$ for any $\stratalt\in\strats$ and (ii) if  $\strat^T\mat\strat= \stratalt^T\mat\strat$, then $\strat^T\mat\stratalt> \stratalt^T\mat\stratalt$.

Finally, a population state $\bar\strat=(\bar x_1,\bar x_2)\in \strats$ is \emph{globally attracting} if for any population state  $\strat=(x_1(0),x_2(0))\in \strats$ such that $(x_1(t),x_2(t))\in \strats$ for any $t\in[0,+\infty)$  there is  
 \[\lim_{t\to+\infty}\limits\rho\left((x_1(t),x_2(t)),(\bar x_1,\bar x_2)\right)=0,\] where $\rho$ is the Euclidean metric in $\mathbb {R}^2$.
 
Thus, if the system has a unique globally attracting evolutionarily stable state, then it can be seen as a limiting state of the whole system.

{\bf A game.} In this setting, we consider a random matching scenario induced by a $2\times 2$ symmetric game with a payoff matrix of the form
\bv
\begin{equation}
\label{eq:game-gen}
\begin{array}{l|cc}
	&s_1	&s_2\\
\hline
s_1	&\gaina		&\gainb\\
s_2	&\gainc	& \gaind
\end{array}
\end{equation}
\bc
where $\gaina,\gainb,\gainc,\gaind\in\mathbb{R}$ and $\gaina\neq \gainc$, $\gaind\neq \gainb$. 
Then
\begin{equation}\label{payoffs}
 \pay_1(\strat)=(\gaina-\gainb)\stratx_1+\gainb;\qquad \pay_2(\strat)=(\gainc-\gaind)\stratx_1+\gaind.
\end{equation}

In the rest of the paper, we focus mainly on anti-coordination games, i.e., those in which $\gaina <\gainc$ and $\gaind<\gainb$. Then the unique Nash equilibrium of the game is $\mathbf{p}=(p,1-p)\in \strats$, where the \bv ratio \bc of agents playing the first strategy is given by 
\bv
 \begin{equation} \label{eq:eq} 
 \eq=\frac{\gainb-\gaind}{\gainc-\gaina+\gainb-\gaind}.
 \end{equation}
\bc 

\subsection{Game dynamics driven by revision protocols} 
\label{sec:revprot}
The microeconomic foundations of evolutionary game theory are rooted in mass action interpretation, and, more precisely, in the theory of revision protocols.

Referring to the textbook of \cite{San10} for details, suppose that each agent occasionally receives an opportunity to switch actions \textendash\ say, based on the rings of a Poisson alarm clock \textendash\ and, at such moments, the agents reconsider their choices of action by comparing their payoffs with the payoff of the action of a randomly chosen individual in the population.
A \emph{revision protocol} of this kind is typically defined by specifying the \emph{conditional switch rate} $\switch_{\pure\purealt}(\strat)$ at which a revising $\pure$-strategist switches to strategy $\purealt$ when the population is at state $\strat\in\strats$. (In particular, $\switch_{\purealt\pure}(\strat(\time))>0$ if, and only if, $\pay_{\purealt}(\strat(\time))<\pay_{\pure}(\strat(\time))$ and $\switch_{\pure\purealt}(\strat(\time))>0$ if, and only if, $\pay_{\pure}(\strat(\time))>\pay_{\purealt}(\strat(\time))$).
In this case, the population rate of agents switching from strategy $\pure$ to strategy $\purealt$ leads to the inflow-outflow equation of the form
\begin{equation}
\label{eq:in-out}
\dot{\stratx_{\pure}}(\time)
	= 	\sum_{\purealt\neq\pure} \stratx_{\purealt}(\time) \switch_{\purealt\pure}(\strat(\time))
		- \stratx_{\pure}(\time) \sum_{\purealt\neq\pure} \switch_{\pure\purealt}(\strat(\time))	
\end{equation}

 The archetypal example of a revision protocol is the pairwise proportional imitation (PPI) of \citet{Hel92}, which leads to the replicator equation \citep{taylorjonker78}. It is described by the switch rate functions
\begin{equation}
\label{eq:PPI}
\switch_{\pure\purealt}(\strat)
	= \stratx_{\purealt} \pospart{\pay_{\purealt}(\strat) - \pay_{\pure}(\strat)},
\end{equation}
where $[r]_+ = \max\{r,0\}$. 
According to this protocol, a revising agent first observes the action of a randomly selected opponent, so a $\purealt$-strategist is observed with probability $\stratx_{\purealt}$ when the population is in state $\strat\in\strats$.
Then, if the payoff of the incumbent strategy $\pure\in\pures$ is lower than that of the benchmark strategy $\purealt$, the agent imitates the selected agent strategy with probability proportional to the payoff difference $\pay_{\purealt}(\strat) - \pay_{\pure}(\strat)$.
Otherwise, the revising agent skips the revision opportunity and sticks to their current action.

Other well-known examples of revision protocols are direct protocols. For example, with switch rates given by   
\begin{equation} \label{eq:Smithprot}\switch_{\pure\purealt}(\strat)= \pospart{\pay_{\purealt}(\strat) - \pay_{\pure}(\strat)}. 
\end{equation}
In this case, agents compare their payoffs with those resulting from other strategies (used or not) and decide to change them only when the payoff from the benchmark strategy is more profitable, which leads to the Smith dynamics \citep{Smi84}. 
Another famous game dynamics is that of Brown, von Neumann and Nash \citep{brown1950solutions,hofbauer2009brown}, determined by the revision protocol of the form:
\begin{equation} \label{eq:BNNprot}\switch_{\pure\purealt}(\strat)= \pospart{\pay_{\purealt}(\strat) - \pay(\strat)}, 
\end{equation}
where $\pay(\strat)$ is the average payoff in $\strat$.
For more information on the revision protocols \bv and dynamics they introduce \bc we refer the reader to \cite{San10} and a more recent article by \citet{MV22}. 

\subsection{Evolutionarily stable states in population game with 2 strategies} 
An evolutionarily stable state may not exist in an arbitrary population game. Nevertheless, in the case of population games with two strategies, we have a simple characterization of ESS's. 

\begin{theorem} \label{thm:ESS2strategies}
    For any population game with a payoff matrix given by \eqref{eq:game-gen}, there exists an evolutionarily stable state (ESS). If $a<c$ and $b>d$, then the game has a unique (\bb interior\bc) ESS 
    of the form 
    \begin{equation}
        \label{poliESS} \strat^*=(\stratx_1^*,\stratx_2^*)=(p,1-p)= \left(\frac{d-b}{a-c+d-b},\frac{a-c}{a-c+d-b}\right).
    \end{equation}
    Otherwise, at least one of the monomorphic states: $(1,0)$ or $(0,1)$ is evolutionarily stable. \medskip
    
    Moreover, for game dynamics given by equation \eqref{eq:in-out}:
\begin{enumerate}
    \item the population state $(1,0)$ is globally attracting ESS when $a>c$ and $b>d$;
    \item the population state $(0,1)$ is globally attracting ESS when $a<c$ and $b<d$;
    \item the \bv interior \bc state $\strat^*$ is globally attracting ESS, if only $a<c$ and $b>d$;
    \item if $a>c$ and $b<d$, then to which ESS the population is attracted depends on the initial state of the population: if $\stratx (0)<p$, then the state of the population is attracted to $(0,1)$, whereas if $\stratx (0)>p$, then it is attracted to $(1,0)$. 
\end{enumerate}    
\end{theorem}
Theorem \ref{thm:ESS2strategies} is a compilation of various well-known facts about population games, 
 taking into account the dynamics of economic agents' behavior (see, e.g. \citep{herold2020evolution,bielawski2025emergence}). For completeness \bv of the exposition\bc, the proof is \bv provided \bc in the Appendix. 
 
 The following consequence of Theorem \ref{thm:ESS2strategies} turns out to be particularly useful in our considerations.

\begin{corollary} \label{cor:polymorphic}
    For any game with payoffs given by \eqref{eq:game-gen}, where $a<c$ and $b>d$, the game dynamics \eqref{eq:in-out} converge to the state $\strat^*$ given by \eqref{poliESS}, for any \bv interior \bc initialization. 
\end{corollary}

This corollary allows us to focus on the unique and evolutionarily stable \bv interior \bc state of the anti-coordination game. Therefore, we can examine the impact of changes in innovation-induced diversity on the economic system by observing how the \bv interior \bc ESS is changing.  In the next section these results will let us introduce concepts of a Schumpeterian game and a Schumpeterian state -- a population state where imitators and innovators coexist.

\section{A Schumpeterian game}
Let us consider a population game in which each firm can choose to use either an innovative technology or a proven one. We will analyze how the dissimilarities between these two technologies (which determine a firm's payoff) impact firms' decisions whether to innovate or not.

\subsection{A relevant innovation}
Consider a product market, with firms modeled by the unit interval $\players$, a set of possible technologies (products)  $\plans$, and a function of diversity $\dive\colon 2^{\plans}\rightarrow [0,+\infty)$. Let $\plan \in\plans$ be some known \emph{old} technology. 

We call $\plan'\in \plans$ a \emph{relevant innovation} of $\plan$, if the diversity of $\plan'$ is greater than the diversity of $y$, that is, \begin{equation}\label{eq:r}
\dive(\{\plan'\})>\dive(\{\plan\}).   
\end{equation} 
Thus, according to equation \eqref{przyrost}, $\plan'\in \plans$ is a relevant innovation of $\plan\in \plans$, if  \begin{equation}\label{różnice}
 \distinct(\plan',\plan)>\distinct(\plan,\plan')\geq 0.
 \end{equation}
Inequality $\distinct(\plan',\plan)>0$ indicates that there is a relevant attribute that the old technology does not have but that the innovation provides.
On the other hand, inequality $\distinct(\plan,\plan')>0$ expresses that if the innovation is the only technology used, some consumers are disappointed by the disappearance of the attributes of the old technology. If innovation possesses all the relevant attributes of the old technology, then $\distinct(\plan,\plan')= 0$.

\subsection{The value of the market with a relevant innovation} We assume that the set of available technologies $\plans$ consists solely of two technologies: $\plans=\{\plan,\plan'\}$, where $\plan$ represents the old technology and $\plan'$ is its relevant innovation. (We make this simplification for clarity of the exposition and to focus on the impact of innovation (dissimilarity in technologies) on the population state. The ideas presented in this section can easily be extended to larger sets of available technologies. Note also that $\plan$ ($\plan'$) can refer to more than one technology that possesses the same attributes.)

Denote by $\improf$ the \bv (aggregated) \bc market value for the old technology $\plan$. \bv We assume $\improf$ is known and constant. \bc Let $\improf'$ be the \bv (aggregated) value of the market with the (relevant) innovation $\plan'$, \bc where $\improf'=\improf+\aprof$, and $\aprof>0$ determines the increase in the market value resulting from the occurrence of innovation. We assume that the value of $\aprof$ depends only on the dissimilarity of the product $\plan'$ from the product $\plan$, that is $\distinct(\plan',\plan)$  (see \eqref{dis}), that is, there is an increasing function $\marketv\colon [0,+\infty)\rightarrow [0,+\infty)$ with $\marketv(0)=0$, such that \begin{equation} \label{ap}\aprof:=\marketv(\distinct(\plan',\plan)).\end{equation}
To indicate the value of the market with innovation, but without old technology, we use a parameter $\bprof$, which depends on the value of $\distinct(\plan,\plan')$ in such a way that \begin{equation}    \label{bp} \bprof:=\marketvu(\distinct(\plan,\plan')), \end{equation} where $\marketvu\colon [0,+\infty)\rightarrow [0,+\infty)$ is an increasing function such that $\marketvu(0)=0$.
The value of the market with $\plan'$ being the only technology used is equal to $\improf+\aprof-\bprof$.  
Since, by \eqref{eq:r}, diversity increases due to innovation, we assume that the value of the market with innovation is also greater than without it, that is, \begin{equation} \label{ab}    \aprof>\bprof\geq 0.\end{equation} 
\bv In particular, if $\marketvu=\marketv$, then it is true for any relevant innovation.

Although both $\alpha$ and $\beta$ are described by dissimilarity, they differ (usually $d(y',y)\neq d(y,y')$). In particular, the change in $\alpha$ is determined by the change of $d(y',y)$, thus determining how much {\it new/additional} value the innovation has. The change in $\beta$ is caused by the value lost due to innovation. 

\bc

\subsection{The Schumpeterian strategies}
We identify a firm strategy with the choice of a specific technology from the set of available technologies $\plans=\{\plan,\plan'\}$, where $\plan$ represents the old technology, and $\plan'$ is its relevant innovation, that is, condition \eqref{różnice} is met. Then, each firm $\pure \in \players$ can choose one of the following two actions:
\begin{enumerate}
    \item [\textbf{$\imit$}] a strategy of refraining from innovation and using old technology $\plan$; this strategy does not entail any additional costs and may generate a profit not exceeding the market value $\improf$ of the old technology;

    \item [\textbf{$\inno$}] an innovation strategy that involves investing a fixed cost of $\cost>0$ in exploiting a new business opportunity (the value of $\cost$ measures the cost of research and development of innovation $\plan'$), which allows the firm to innovate and achieve a profit at most equal to $
     \improf+\aprof-\bprof-\cost$,
  with $\aprof$ and $\bprof$ given by \eqref{ap} and \eqref{bp} respectively.
    \end{enumerate}  
Since actions $\inno$ and $\imit$ refer to the specific behavior of the firm described by \cite{Schumpeter2017TheoryCycle} we refer to the set $S=\{ \inno, \imit\}$ as \emph{the set of Schumpeterian strategies}.

When looking into a firm's payoffs in the two-agent symmetric game with Schumpeterian strategies, we consider the following three matchings.
\begin{enumerate}
    \item If both firms refrain from innovating and use the old technology, the profit $\improf$ is divided equally between both and each firm receives a payoff equal to $$\pay_1(\imit,\imit )=\pay_2(\imit,\imit )=\frac{\improf}{2}.$$
    
    \item If both firms decide to innovate, which means that each invests $\cost>0$ in developing the new technology, the profit $\improf+\aprof-\bprof$ is equally divided between both firms; thus, each firm receives a payoff equal to $$\pay_1(\inno,\inno )=\pay_2(\inno,\inno )=\frac{\improf+\aprof-\bprof}{2}-\cost.$$ 

    \item If one firm implements an innovation, while the other uses the old technology, the innovator, who pays the cost of innovation $\cost>0$, gains $\improf+\aprof-\bprof-\cost$; while the payoff for the firm that does not innovate is equal to $\bprof$:
$$\pay_1(\inno,\imit)=\pay_2(\imit,\inno)=\improf+\aprof-\bprof-\cost,\;\;\;\; \;\;\;\;\pay_1(\imit,\inno)=\pay_2(\inno,\imit)=\bprof.$$  
 \end{enumerate}
 Combining the above, the payoff matrix of the $2\times 2$ symmetric game with strategies of the set $S=\{ \inno, \imit\}$ takes the form
 \bv
  \begin{equation}
\label{eq:sgame-gen}
\begin{array}{l|cc}  
	&\inno	&\imit\\
\hline 
\inno	&\frac{\improf+\aprof-\bprof}{2}-\cost &\improf+\aprof-\bprof-\cost\\
\imit	&\bprof	& \frac{\improf}{2}
\end{array},\vspace{20pt}
\end{equation}
\bc
with $\improf>0$, $\cost>0$, and $\aprof>\bprof\geq 0$.

The $-\bprof$ component of the innovator's profit in scenario (2) indicates that a certain relevant attribute of the old technology is no longer available in the market when both firms decide to innovate. 
However, scenario (3) offers a different interpretation of the parameter $\bprof$, as it determines the non-innovating firm's share of the market value with innovation. 
\begin{remark}
     Despite many possible ways of dividing profits, we assume that firms using the same strategy share profits equally, which fits the symmetric random matching model. It typically imposes that every player is statistically identical in terms of the probability of being matched with another player, and thus, a uniform matching probability across the population.
 \end{remark}

\bv
\begin{remark}
 The decision whether to innovate is usually intertwined with uncertainty about how the innovation will be received. In the presented model, this uncertainty may have an impact on  $\aprof$ and $\bprof$. We assume that they can be precisely estimated by players. Although this may be seen as a limitation of the model, taking this uncertainty into account (even through an imprecise approximation of the Nash equilibrium) would require a more complex solution like $\varepsilon$-Nash or quantal response equilibria \cite{mckelvey1995quantal}, as well as different dynamics (see, e.g. \cite{bielawski2024memory,CGM,ho2007self,sato2003coupled,hofbauer2005learning,kaniovski1995learning}). Moreover, as in our paper we will focus on global ESS, assuming the perturbations following from uncertainty are small, reasoning on ESS remains appropriate.
 \end{remark}
\bc 

Comment should be made \bv about the consequences of \bc observing payoffs in the game. First, if $\bprof>\frac{\improf}{2}$, and consequently $\pay_1(\imit,\inno)>\pay_1(\imit,\imit)$, then a firm that does not innovate gains a higher payoff when competing with an innovator than with another firm that does not innovate. 
On the other hand, if all relevant attributes of the old technology are captured by innovation, and thus $\bprof=0$, then the non-innovating firm has zero profit when meeting an innovator. However, even in this case, firms can still avoid choosing the innovating strategy if the cost of innovation is high enough, that is, if $\cost>\improf+\aprof$, which results in $\pay_1(\inno,\imit)<\pay_1(\imit,\imit)$. 
Looking at the other extreme, when $\bprof=\improf$, if both firms implement innovation, that is, if $\pay_1(\inno,\inno)>\pay_1(\imit,\inno)$, then the increase in the market value resulting from innovation must be sufficiently large compared to the cost of innovation ($\aprof>2\cost$).

\subsection{A Schumpeterian game and the Schumpeterian state} A random matching scenario induced by the two-agent symmetric game with payoff matrix of the form \eqref{eq:sgame-gen} generates a population game in which, at each time, the firms in a given market choose whether to innovate or not.

The population of firms that choose Schumpeterian strategies from the set $S=\{ \inno, \imit\}$ forms a population state $\strat=(\stratx_I,\stratx_R)\in\strats$, in which $\stratx_I \in [0,1]$ denotes the rate of \emph{innovators}, i.e., firms choosing strategy $I$, and $\stratx_R=1-\stratx_I$ is the rate of \emph{imitators}, which refrain from innovation using strategy $R$. The payoff vector $\payv(\strat) = (\pay_{I}(\strat),\pay_{R}(\strat))$ consists of the payoffs for the innovator \begin{equation}\label{payoff1}
 \pay_I(\strat)=\improf+\aprof-\bprof-\cost-\frac{\improf+\aprof-\bprof}{2}\stratx_I,\end{equation}  and the imitator  \begin{equation}\label{payoff2} 
 \pay_R(\strat)=\left(\bprof-\frac{\improf}{2}\right)\stratx_I +\frac{\improf}{2}.
\end{equation}

Obviously, the gain from each strategy depends on the rate of innovators (and consequently imitators). In particular, by \eqref{payoff1}, being an innovator is more beneficial the lower the fraction of innovators on the market. On the other hand, according to \eqref{payoff2}, whether the imitator's payoff increases or decreases as the rate of innovators rises depends on the relationship between the imitator's share of the market with innovation $(\bprof)$ and without innovation $(\frac{\improf}{2})$. Specifically, the larger the fraction of innovators in the market, the payoff of the imitator increases if $\bprof>\frac{\improf}{2}$, decreases if $\bprof<\frac{\improf}{2}$, and is constant, equal to $\frac{\improf}{2}$ if $\bprof=\frac{\improf}{2}$.

\bb 
The aggregated market value without innovation, denoted by $\improf$ 
is 
the total profit gained by the firms in the given market. 
Thus, the value of $\improf$ is equal to the integral over the set of producers $\consumers$ of a measurable function that assigns profit to each firm (compare to \eqref{02}). 
 This specification, interlocked with the way agents interact through symmetric random matching in a population game, 
could be interpreted as firms competing in local, bilateral "representative"  interactions.
In this interpretation, the payoffs measure relative success in a given interaction. An important
implication is that the payoff from such an interaction does not depend on the overall
market share of innovators. The market share of innovators determines the expected (average) payoff in the population. Then the success (or lack thereof) of the player's chosen strategy influences their subsequent choices. This is reflected by the revision protocols. Therefore, these choices drive the dynamics of the game.\bc

 In the Schumpeterian approach to economic development \citep{Schumpeter2017TheoryCycle}, the market accommodates both innovators introducing new technological solutions and firms that refrain from engaging in innovation. Then, when considering Schumpeterian competition, we are interested in a population game that leads to a stable market structure, consisting of both innovators and non-innovators. 
This means that we aim for a \bv mixed \bc and evolutionarily stable population of firms, with a positive rate of firms that innovate and those that refrain from innovation. To strengthen the message, this Schumpeterian market structure will be the result of decisions made by firms governed by a natural rule --- revision protocol (as described in Section \ref{sec:revprot}). Therefore, we focus on those population games in which dynamics cause the population of firms to converge (in the long run) to a unique, globally attracting, evolutionarily stable state, regardless of the initial state of the population. Specifically, we can focus on population games that have the unique \bv interior \bc Nash equilibrium of the form $\strat^*=(\stratx_I^*,\stratx_R^*)\in \strats$ given by \eqref{poliESS}.   
By Corollary \ref{cor:polymorphic}, this state is evolutionarily stable only if 
\bv
\begin{equation}
\label{sch}\frac{\improf}{2}+\frac{\aprof-\bprof}{2}<\bprof+\cost \text { and } \frac{\improf}{2} +\aprof>\bprof+\cost.\end{equation}
\bc
\begin{definition}
We call a population game $\gamefull$ defined on the set of Schumpeterian strategies $S=\{ \inno, \imit\}$ and the payoff matrix given by \eqref{eq:game-gen} \emph{a Schumpeterian game} if condition \eqref{sch} is met. 
\end{definition} 

A Schumpeterian game is a population game in which firms decide whether to engage in innovative activities, with firms' payoffs including the impact of a market quality change related to the increase in diversity coming from innovation. The condition \eqref{sch} imposed on the payoffs and, consequently, on the parameters of the game ensures that there is a unique state that the firm population is approaching in the long term. Specifically, by Theorem \ref{thm:ESS2strategies} and Corollary \ref{cor:polymorphic}, we have the following observation.

\begin{proposition}\label{prop:sstate}
In a Schumpeterian game $\gamefull$, there is a unique \bv interior \bc globally attracting ESS of the form $\strat^*=(\stratx_I^*,\stratx_R^*)\in \strats$, hereinafter referred to as \textbf{the Schumpeterian state}, in which the rate of innovators is equal to 
\begin{equation}\label{q12}
    \stratx_I^*=\frac{\improf+2(\aprof-\bprof-\cost) }{\aprof+\bprof},\color{black}
\end{equation} 
and the rate of imitators is given by 
\begin{equation}\label{q12a}
 \stratx_R^*=\frac{2\cost-\aprof-\improf+3\bprof}{\aprof+\bprof}.  \end{equation}
Moreover, the game dynamics \eqref{eq:in-out} converges to the Schumpeterian state $\strat^*$ for any \bv interior \bc initialization. 
\end{proposition}

The Schumpeterian state is the final state of the population game in which firms decide whether to innovate or not. It exists and is unique, by Proposition \ref{prop:sstate}, in any Schumpeterian game. In the Schumpeterian state, a population consists of innovators who coexist on the market with firms that do not innovate. The innovator and imitator rates are fixed and depend on the market parameters, which include the cost of innovation ($\cost$), the original market value ($\improf$), and the changes in the market caused by innovation, measured by $\aprof$ and $\bprof$. It follows from the character of the indexes \eqref{q12} and \eqref{q12a} that, ceteris paribus, the higher the cost of innovation, the fewer firms decide to innovate (and more refrain from innovation). On the other hand, the higher the original market value $\improf$, the more innovators and fewer imitators in the market. The concept of the Schumpeterian state that we use corresponds to the definition of an innovation system used by \cite{Andersen2007BridgingTheory}.

Starting from any \bv interior \bc population state, evolution will bring the population close to an equilibrium in which the ratio of both types of agents is constant. Although firms can change their action choices (according to an adopted revision protocol) in the Schumpeterian state, the overall proportion of firms adopting a given strategy cannot change unless market parameters change.

\section{Diversity vs. innovation in a Schumpeterian game} \label{ESS} 
In our model, Schumpeterian competition leads to a stable market structure, called the Schumpeterian state. In such a market, firms using proven technological solutions and firms developing innovations operate together, with rates determining the composition of the firm population set by the game parameters.
Thus, by Proposition \ref{prop:sstate}, to study how the change in diversity caused by innovation affects the fraction of innovators (and imitators) in the population, we can focus on the Schumpeterian state, and we can conclude what the impact of innovation-induced diversity change on firm strategic choices is. 

There are four parameters of a Schumpeterian game. Two of them, not related to the change in diversity, are the original market value $\improf$ and the cost of innovation $\cost$, both of which depend on the market environment, which could be modified through technology policy. Two parameters directly related to innovation are $\alpha>0$, which values growth, and $\beta\geq 0$, which values the reduction in diversity associated with innovation. Those parameters depend mainly on the firm's ability to determine and satisfy consumer needs. 

To investigate how the values of $\alpha$ and $\beta$ affect the Schumpeterian state, specifically the fraction of innovators \eqref{q12}, we denote $\xi:=2\cost-\improf$, and assume it is fixed. Condition \eqref{sch} imposes that 
$\xi\in \left(\aprof-3\bprof,2(\aprof-\bprof)\right)$. Then the rate of innovators \eqref{q12} in the Schumpeterian state is given by 
\begin{equation}\label{gamma}
\stratx_I^*=\frac{2(\aprof-\bprof)-\xi}{\aprof+\bprof}.
\end{equation}
Simple calculations (see the proof in Appendix) lead to the following proposition, which describes the effect of changes in parameters $\alpha$ and $\beta$ on the rate of innovators in the Schumpeterian state.

\begin{proposition}\label{p:rate}
In the Schumpeterian state $\strat^*=(\stratx_I^*,\stratx_R^*)$, the rate of innovators $\stratx_I^*$ increases if $\aprof$ increases (with $\xi$ and $\bprof$ fixed) or $\bprof$ decreases (with $\xi$ and $\aprof$ fixed).
   \end{proposition}

\subsection{The increase in diversity and the fraction of innovators in the Schumpeterian state}
The change in diversity related to the replacement of the old technology $\plan$ by its relevant innovation $\plan'$ is measured by the difference $\dive(\{\plan'\})-\dive(\{\plan\})>0$ (see: \eqref{eq:r}). However, in the Schumpeterian state, by Proposition \ref{prop:sstate}, both Schumpeterian strategies are used, and together with firms that introduce innovation, also those that stick to the old technology operate on the market. Since both technologies are used, the attributes of both are available, and the diversity is not reduced. What measures the diversity change occurring from innovation in the Schumpeterian state is therefore the dissimilarity of innovation $\plan'$ from the old technology $\plan$ given by $\distinct(\plan',\plan)$, which indicates how much consumers value the attributes of innovation compared to those of the old product. 
The bigger the dissimilarity  $\distinct(\plan',\plan)$, the higher the value of parameter $\aprof$ (see \eqref{ap}), which determines the increase in the market value resulting from the appearance of the innovation. That means that the better the innovation from the point of view of consumers, the higher $\aprof$. On the other hand, by Proposition \ref{p:rate}, growing $\alpha$ increases the rate of innovators in the Schumpeterian state. Therefore, we obtain the following conclusion.

\begin{corollary}
  \label{prop:s1}  
The greater the increase in diversity caused by innovation, the more innovators are in the Schumpeterian state.
\end{corollary}

\bb

Since the Schumpeterian state determines a mixed Nash equilibrium of a population game, the payoffs from both strategies used by firms in the Schumpeterian state $\strat^*$ are equal. 
The payoff for the population is an affine function of the ratio of innovators $\stratx_I^*$, and its monotonicity depends on the relationship between the share of the imitators in the market with innovation measured by $\bprof$ and the aggregated value of the market without innovation ${\improf}$. By \eqref{payoff2}
\begin{equation}\label{payoffp}  
    u(\strat^*)=u(\stratx_I^*,\stratx_R^*)=\stratx_I^*\left(\bprof-\frac{\improf}{2}\right)+\frac{\improf}{2}. 
\end{equation}

Thus, when $\bprof < \frac{\improf}{2}$, an increase in the ratio of innovators results in a lower population payoff. 
Moreover, in the extreme case of $\beta=0$, the growth of the ratio of innovators leads the population's payoff to zero. 
 However, when the inequality $\beta>\frac{V}{2}$ holds, the population payoff grows with the increasing mass of innovators. In this case, innovators exploiting a certain market niche enable imitators to take over the market of the old good, making all players benefit from the innovation. Hence, (anti)coordination of firms' activities is crucial for the creation of the Schumpeterian state. 
  Therefore, we can supplement Corollary \ref{prop:s1} with the following observation. 

\begin{corollary}
  \label{cor:imit} Provided that the imitators' share in the market is large enough, the more innovators in the Schumpeterian state, the higher the population payoff.

\end{corollary}
 It can be noted that the above-mentioned division of Schumpeterian games, depending on the relationship between parameters $\bprof$ and  ${\improf}$, corresponds to one of possible general divisions of population games: conflict games and anti-coordination games, see \cite{herold2020evolution}.


\bc

Considering Schumpeterian competition as a population game in which firms decide whether to invest in innovation or to abandon the idea, we conclude that diversity determines the result of the game. Not only does it affect the structure of the market that emerges as a result of competition, but it also determines whether there is any stable market structure at all (see condition \eqref{sch}). 
Our approach to diversity is based on the consumers' values for the different product attributes. The firm that meets the consumers' expectations regarding the new product or technology can temporarily gain a high extra payoff, which is then reduced by others learning to do the same by applying the given revision protocol (described in Section 3.1). In the long run, firms' decisions may lead some markets to a Schumpeterian state. 
Its form depends on external conditions that decide the cost of innovation and the original market value; however, it also depends on the firm that develops the innovation itself. Eventually, the more valuable an innovation a firm develops, the more firms will follow it, and the final form of the market will consist of more new technology-using firms.

\section{Conclusions}

In this paper, we analyze the impact of a change in diversity introduced by a new product on the evolution of the economic system. To this aim, we considered a Schumpeterian game in which firms decide on a specific Schumpeterian strategy, and the innovation-induced diversity change (measured by dissimilarity of products) modifies the structure of the agent population by affecting the firms' strategic choices. We showed that innovations that meet consumers needs by increasing economic diversity can create a stable and attracting population structure, which is an evolutionarily stable state.
In this state (called the Schumpeterian state), innovators and firms that use proven technological solutions (imitators) coexist. In addition, the fraction of innovators depends on the change in diversity (and dissimilarity of a new product) associated with innovation. The greater the increase in diversity that innovative technology provides, compared to the diversity of the technology previously used, the more firms decide to develop new technologies.

By analyzing the impact of increasing diversity on the structure of firms in a Schumpeterian game, we have shown that firms are more likely to introduce new technologies that offer greater diversity growth. This implies that the increase in diversity caused by innovation can be seen as an encouragement for firms to innovate.(Moreover, in reality, one usually observes only how many (or what fraction of) firms decide to innovate, and thus analyzing Schumpeterian states can shed light on the real costs and profits associated with innovation). Thus, if an innovation can be introduced but its cost is high, the change in diversity (the dissimilarity) of the new product from  previously used solutions may be a motivation for its introduction. Therefore, the high cost of innovation can be mitigated by its relevance. This is consistent with the perspective of many pro-ecological solutions introduced in response to climate change, etc.
\bv Thus\bc, we argue that product diversity should be considered one of the elements influencing the process of creative destruction, also in Schumpeterian growth models, which can be viewed as infinite sequences of Schumpeterian games.

\bv 
\subsection{More strategies.} We identified the strategy space of the Schumpeterian game with the set consisting of two strategies. From the economic perspective, the reason was that we focused on the interaction in the population between innovators and imitators and their motivations to choose to innovate or not. Nevertheless, extensions of the model may be considered in various directions. What if we add a new strategy? Consider a third strategy: the retaliator (counterattack) strategy ($\reta$). A player using this strategy behaves like an innovator when he meets an innovator and retracts from innovation if the opponent does not innovate. (This game is motivated by Hawk-Dove-Retaliator game (see \cite{bomze1983lotka,kleshnina2022shifts}).) Then the payoff matrix of this game with strategies $S=\{ \inno, \imit, \reta \}$ takes the form
  \begin{equation}
\label{eq:sgame-gen3}
\begin{array}{l|ccc}  
	&\inno	&\imit& \reta\\
\hline 
\inno	&\frac{\improf+\aprof-\bprof}{2}-\cost &\improf+\aprof-\bprof-\cost &\frac{\improf+\aprof-\bprof}{2}-\cost \\
\imit	&\bprof	& \frac{\improf}{2}& \frac{\improf}{2} \\
\reta & \frac{\improf+\aprof-\bprof}{2}-\cost & \frac{\improf}{2}& \frac{\improf}{2}
\end{array},
\end{equation}
with $\improf>0$, $\cost>0$, and $\aprof>\bprof\geq 0$.

Even in such a simple case, the behavior can be significantly different. Consider the numerical example of \eqref{eq:sgame-gen3} with $V=8, \alpha=4,\beta=2$ and $C_I=4$. Then the game is given by 

  \begin{equation}
\label{eq:sgame-gen3num}
\begin{array}{l|ccc}  
	&\inno	&\imit& \reta\\
\hline 
\inno	&1 &6 &1 \\
\imit	&2	& 4& 4 \\
\reta & 1 & 4& 4
\end{array},
\end{equation}

In Figure \ref{fig:fig1} we present the phase portrait for game \eqref{eq:sgame-gen3num}, where players are using pairwise proportional imitation revision protocol \eqref{eq:PPI}. We see that
\begin{enumerate}
    \item the Schumpeterian state from 2-strategy game \eqref{eq:game-gen} (by Proposition \ref{prop:sstate} given by $(\stratx_{\inno}^*,\stratx_{\imit}^*)=(2/3,1/3)$) sets the stable solution of this extended game.
    \item Although the population state $(x_I^*,x_R^*,0)$ attracts, it only attracts locally.
    \item There exist initial population states, for which the trajectories (driven by replicator dynamics) will converge to the population state with no innovators -- consisting only of imitators and retaliators.
\end{enumerate}

\begin{figure}[h!]
\begin{center}
\includegraphics[width=110truemm]{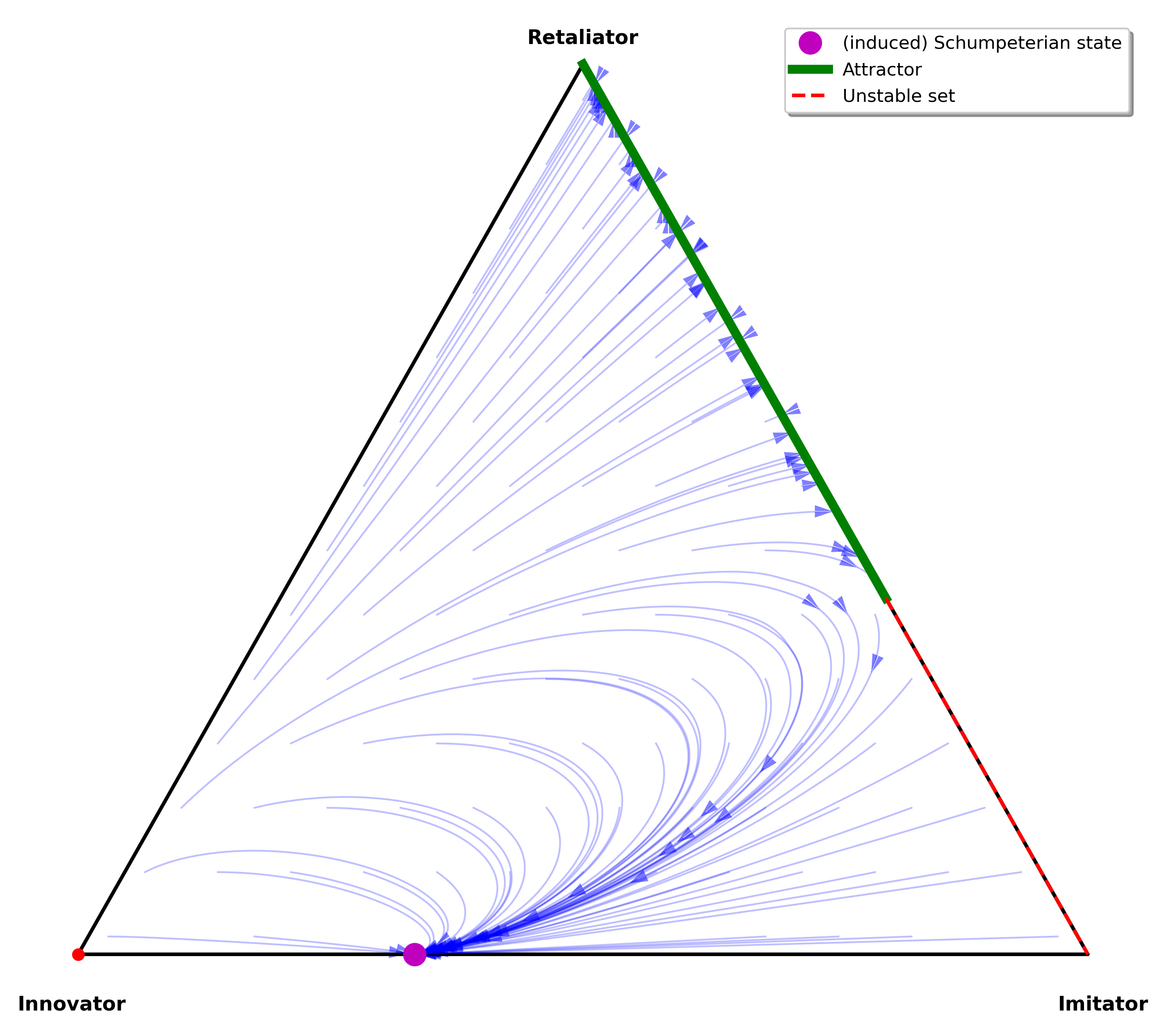}
\caption{Phase portrait of the game \eqref{eq:sgame-gen3num} where agents are using pairwise proportional imitation \eqref{eq:PPI}. Schumpeterian state of the game \eqref{eq:sgame-gen} $\strat^*=(\stratx_{\inno}^*,\stratx_{\imit}^*)$ determines ESS of game \eqref{eq:sgame-gen3num} given by $\overline{\strat}=(\stratx_{\inno}^*,\stratx_{\imit}^*,0)$. In contrary to the 2-strategy case (see Proposition \ref{prop:sstate}) $\overline{\strat}$ is not globally attracting. Part of the simplex is attracted by a stable section consisting of imitators and retaliators.} \label{fig:fig1}
\end{center}
\end{figure}

We can extend this reasoning as long as assumptions of Proposition \ref{prop:sstate} are met (that is, as long as there exists an interior ESS of the game \eqref{eq:sgame-gen}). Interestingly, population state $\overline{\strat}=(x_I^*,x_R^*,0)$, which coincides with the Schumpeterian state determined in Proposition \ref{prop:sstate}, is a stable solution of the extended game \eqref{eq:sgame-gen3}. Nevertheless, it is no longer globally attracting. 
Moreover, by \cite{bomze1983lotka,San10,sandholm2010local}, we can infer that the population state $\overline{\strat}$, where ratios of innovators and imitators coincide with the Schumpeterian state $\strat^*$  (and no player is a retaliator), locally attracts for a large class of revision protocols (in particular under the best response rule).  Nevertheless, for a large set of initial conditions, the trajectories will converge to the population state consisting solely of imitators and retaliators. Eventually, no innovators will be observed then.  

 
 \subsection{Future work} \bc This article serves as a direct interconnection between the relevance of innovation and the diversity of a new product. In future work, one may take into account the impact of diversity (and innovation) on market entry and market structure itself, as well as the external effects of innovation. The proposed model of a Schumpeterian game can be extended in the context of stable games (see \citep{hofbauer2009stable}), games characterized by self-defeating externalities. 
\bv 
Finally, one can ask what happens when diversity is generated endogenously through the strategic choice of attributes. In particular, how the system evolves (and how payoffs in population change) when a new attribute is added. In this context, as the set of available strategies changes, one should use a different approach, see e.g. \cite{edhan2025game}. As the game dynamics changes one can consider then potential games and gradient descent arguments for convergence. \bc

{\bf Acknowledgments.} Elżbieta Pliś and Fryderyk Falniowski acknowledge the support of the National Science Centre, Poland, grant 2023/51/B/HS4/01343. We are indebted to Agnieszka Lipieta and M. Ali Khan for many fruitful discussions and encouragement.

\bibliographystyle{abbrvnat} 
\bibliography{ms,IEEEabrv,Bibliography,ela}

\begin{thebibliography}{65}
\providecommand{\natexlab}[1]{#1}
\providecommand{\url}[1]{\texttt{#1}}
\expandafter\ifx\csname urlstyle\endcsname\relax
  \providecommand{\doi}[1]{doi: #1}\else
  \providecommand{\doi}{doi: \begingroup \urlstyle{rm}\Url}\fi

\bibitem[Acemoglu(2008)]{Acemoglu2008IntroductionGrowth}
D.~Acemoglu.
\newblock \emph{{Introduction to modern economic growth}}.
\newblock Princeton University Press, 2008.

\bibitem[Acemoglu et~al.(2023)Acemoglu, Ozdaglar, and Pattathil]{acemoglu2023learning}
D.~Acemoglu, A.~Ozdaglar, and S.~Pattathil.
\newblock Learning, diversity and adaptation in changing environments: The role of weak links.
\newblock Technical report, National Bureau of Economic Research, 2023.

\bibitem[Aghion and Howitt(1992)]{Aghion1992ADestruction}
P.~Aghion and P.~Howitt.
\newblock A model of growth through creative destruction.
\newblock \emph{Econometrica}, 60\penalty0 (2):\penalty0 323--351, 1992.

\bibitem[Aghion and Howitt(2009)]{Aghion2009}
P.~Aghion and P.~Howitt.
\newblock \emph{The Economics of Growth}.
\newblock MIT Press, 2009.

\bibitem[Aghion and Howitt(2022)]{aghion2022creative}
P.~Aghion and P.~Howitt.
\newblock Creative destruction and {US} economic growth.
\newblock \emph{Capitalism \& Society}, 16\penalty0 (1), 2022.

\bibitem[Aghion and Howitt(2023)]{aghion2023creative}
P.~Aghion and P.~Howitt.
\newblock The creative destruction approach to growth economics.
\newblock \emph{European Review}, 31\penalty0 (4):\penalty0 312--325, 2023.

\bibitem[Aghion et~al.(2019)Aghion, Akcigit, Bergeaud, Blundell, and H{\'e}mous]{aghion2019innovation}
P.~Aghion, U.~Akcigit, A.~Bergeaud, R.~Blundell, and D.~H{\'e}mous.
\newblock Innovation and top income inequality.
\newblock \emph{The Review of Economic Studies}, 86\penalty0 (1):\penalty0 1--45, 2019.

\bibitem[Amir and Wooders(2000)]{RePEc:eee:gamebe:v:31:y:2000:i:1:p:1-25}
R.~Amir and J.~Wooders.
\newblock One-way spillovers, endogenous innovator/imitator roles, and research joint ventures.
\newblock \emph{Games and Economic Behavior}, 31\penalty0 (1):\penalty0 1--25, 2000.

\bibitem[Andersen(2007)]{Andersen2007BridgingTheory}
E.~S. Andersen.
\newblock {Bridging the gap between Schumpeterian competition and evolutionary game theory}.
\newblock In \emph{DRUID Summer Conference 2007}, Copenhagen, 2007.

\bibitem[Anderson and Renault(1999)]{anderson1999pricing}
S.~P. Anderson and R.~Renault.
\newblock Pricing, product diversity, and search costs: A {B}ertrand-{C}hamberlin-{D}iamond model.
\newblock \emph{The RAND Journal of Economics}, pages 719--735, 1999.

\bibitem[Arbatl{\i} et~al.(2020)Arbatl{\i}, Ashraf, Galor, and Klemp]{arbatli2020diversity}
C.~E. Arbatl{\i}, Q.~H. Ashraf, O.~Galor, and M.~Klemp.
\newblock Diversity and conflict.
\newblock \emph{Econometrica}, 88\penalty0 (2):\penalty0 727--797, 2020.

\bibitem[Barber{\`a} et~al.(2004)Barber{\`a}, Bossert, and Pattanaik]{barbera2004ranking}
S.~Barber{\`a}, W.~Bossert, and P.~K. Pattanaik.
\newblock Ranking sets of objects.
\newblock In \emph{Handbook of Utility Theory: Volume 2 Extensions}, pages 893--977. Springer, 2004.

\bibitem[Belleflamme and Vergari(2011)]{belleflamme2011incentives}
P.~Belleflamme and C.~Vergari.
\newblock Incentives to innovate in oligopolies.
\newblock \emph{The Manchester School}, 79\penalty0 (1):\penalty0 6--28, 2011.

\bibitem[Bielawski et~al.(2024)Bielawski, Chotibut, Falniowski, Misiurewicz, and Piliouras]{bielawski2024memory}
J.~Bielawski, T.~Chotibut, F.~Falniowski, M.~Misiurewicz, and G.~Piliouras.
\newblock Memory loss can prevent chaos in games dynamics.
\newblock \emph{Chaos: An Interdisciplinary Journal of Nonlinear Science}, 34\penalty0 (1):\penalty0 013146, 2024.
\newblock \doi{10.1063/5.0184318}.

\bibitem[Bielawski et~al.(2025)Bielawski, Cholewa, and Falniowski]{bielawski2025emergence}
J.~Bielawski, {\L}.~Cholewa, and F.~Falniowski.
\newblock The emergence of chaos in population game dynamics induced by comparisons.
\newblock \emph{Dynamic Games and Applications}, 15\penalty0 (4):\penalty0 1317--1362, 2025.

\bibitem[Bomze(1983)]{bomze1983lotka}
I.~M. Bomze.
\newblock Lotka-{V}olterra equation and replicator dynamics: a two-dimensional classification.
\newblock \emph{Biological cybernetics}, 48\penalty0 (3):\penalty0 201--211, 1983.

\bibitem[Brown and Von~Neumann(1950)]{brown1950solutions}
G.~W. Brown and J.~Von~Neumann.
\newblock \emph{Solutions of games by differential equations}.
\newblock Rand Corporation, 1950.

\bibitem[Chang(2012)]{chang2012monopolistic}
W.~W. Chang.
\newblock Monopolistic competition and product diversity: Review and extension.
\newblock \emph{Journal of Economic Surveys}, 26\penalty0 (5):\penalty0 879--910, 2012.

\bibitem[Coucheney et~al.(2015)Coucheney, Gaujal, and Mertikopoulos]{CGM}
P.~Coucheney, B.~Gaujal, and P.~Mertikopoulos.
\newblock Penalty-regulated dynamics and robust learning procedures in games.
\newblock \emph{Mathematics of Operations Research}, 40\penalty0 (3):\penalty0 611--633, 2015.

\bibitem[Delbono and Denicolo(1991)]{delbono1991incentives}
F.~Delbono and V.~Denicolo.
\newblock Incentives to innovate in a {C}ournot oligopoly.
\newblock \emph{The Quarterly Journal of Economics}, 106\penalty0 (3):\penalty0 951--961, 1991.

\bibitem[Dixit and Stiglitz(1977)]{dixitst77}
A.~K. Dixit and J.~E. Stiglitz.
\newblock Monopolistic competition and optimum product diversity.
\newblock \emph{The American Economic Review}, 67\penalty0 (3):\penalty0 297--308, 1977.

\bibitem[Edhan and Hellman(2025)]{edhan2025game}
O.~Edhan and Z.~Hellman.
\newblock Game changing mutation.
\newblock \emph{Royal Society Open Science}, 12\penalty0 (4):\penalty0 241951, 2025.

\bibitem[Falniowski(2020)]{falniowski2020entropy}
F.~Falniowski.
\newblock Entropy-based measure of statistical complexity of a game strategy.
\newblock \emph{Entropy}, 22\penalty0 (4):\penalty0 470, 2020.

\bibitem[Ferreira et~al.(2014)Ferreira, Manso, and Silva]{ferreira2014incentives}
D.~Ferreira, G.~Manso, and A.~C. Silva.
\newblock Incentives to innovate and the decision to go public or private.
\newblock \emph{The Review of Financial Studies}, 27\penalty0 (1):\penalty0 256--300, 2014.

\bibitem[Hadikhanloo et~al.(2022)Hadikhanloo, Laraki, Mertikopoulos, and Sorin]{HLMS22}
S.~Hadikhanloo, R.~Laraki, P.~Mertikopoulos, and S.~Sorin.
\newblock Learning in nonatomic games, {Part I}: {Finite} action spaces and population games.
\newblock \emph{Journal of Dynamics and Games}, 9\penalty0 (4, William H. Sandholm memorial issue):\penalty0 433--460, October 2022.

\bibitem[Helbing(1992)]{Hel92}
D.~Helbing.
\newblock A mathematical model for behavioral changes by pair interactions.
\newblock In G.~Haag, U.~Mueller, and K.~G. Troitzsch, editors, \emph{Economic Evolution and Demographic Change: Formal Models in Social Sciences}, pages 330--348. Springer, Berlin, 1992.

\bibitem[Herold and Kuzmics(2020)]{herold2020evolution}
F.~Herold and C.~Kuzmics.
\newblock The evolution of taking roles.
\newblock \emph{Journal of Economic Behavior \& Organization}, 174:\penalty0 38--63, 2020.

\bibitem[Ho et~al.(2007)Ho, Camerer, and Chong]{ho2007self}
T.~H. Ho, C.~F. Camerer, and J.-K. Chong.
\newblock Self-tuning experience weighted attraction learning in games.
\newblock \emph{Journal of {E}conomic {T}heory}, 133\penalty0 (1):\penalty0 177--198, 2007.

\bibitem[Hofbauer and Hopkins(2005)]{hofbauer2005learning}
J.~Hofbauer and E.~Hopkins.
\newblock Learning in perturbed asymmetric games.
\newblock \emph{Games and Economic Behavior}, 52\penalty0 (1):\penalty0 133--152, 2005.

\bibitem[Hofbauer and Sandholm(2009)]{hofbauer2009stable}
J.~Hofbauer and W.~H. Sandholm.
\newblock Stable games and their dynamics.
\newblock \emph{Journal of Economic Theory}, 144\penalty0 (4):\penalty0 1665--1693, 2009.

\bibitem[Hofbauer et~al.(2009)Hofbauer, Oechssler, and Riedel]{hofbauer2009brown}
J.~Hofbauer, J.~Oechssler, and F.~Riedel.
\newblock Brown--von {N}eumann--{N}ash dynamics: the continuous strategy case.
\newblock \emph{Games and Economic Behavior}, 65\penalty0 (2):\penalty0 406--429, 2009.

\bibitem[Jost(2006)]{jost2006entropy}
L.~Jost.
\newblock Entropy and diversity.
\newblock \emph{Oikos}, 113\penalty0 (2):\penalty0 363--375, 2006.

\bibitem[Kaniovski and Young(1995)]{kaniovski1995learning}
Y.~M. Kaniovski and H.~P. Young.
\newblock Learning dynamics in games with stochastic perturbations.
\newblock \emph{Games and Economic Behavior}, 11\penalty0 (2):\penalty0 330--363, 1995.

\bibitem[Kleshnina et~al.(2022)Kleshnina, McKerral, Gonz{\'a}lez-Tokman, Filar, and Mitchell]{kleshnina2022shifts}
M.~Kleshnina, J.~C. McKerral, C.~Gonz{\'a}lez-Tokman, J.~A. Filar, and J.~G. Mitchell.
\newblock Shifts in evolutionary balance of phenotypes under environmental changes.
\newblock \emph{Royal Society Open Science}, 9\penalty0 (11):\penalty0 220744, 2022.

\bibitem[Lancaster(1982)]{LancasterInnovative}
K.~Lancaster.
\newblock Innovative entry: Profit hidden beneath the zero.
\newblock \emph{The Journal of Industrial Economics}, 31\penalty0 (1):\penalty0 41--56, 1982.

\bibitem[Lancaster(1966)]{Lancaster1966ATheory}
K.~J. Lancaster.
\newblock A new approach to consumer theory.
\newblock \emph{The Journal of Political Economy}, 74\penalty0 (2):\penalty0 132--157, 1966.

\bibitem[Lipieta and Pli{\'s}(2022)]{lipieta2022diversity}
A.~Lipieta and E.~Pli{\'s}.
\newblock Diversity and mechanisms of economic evolution.
\newblock \emph{Journal of Evolutionary Economics}, 32\penalty0 (4):\penalty0 1265--1286, 2022.

\bibitem[Magnolfi et~al.(2025)Magnolfi, McClure, and Sorensen]{magnolfi}
L.~Magnolfi, J.~McClure, and A.~Sorensen.
\newblock {Triplet Embeddings for Demand Estimation}.
\newblock \emph{American Economic Journal: Microeconomics}, 17\penalty0 (1):\penalty0 282--307, 2 2025.
\newblock ISSN 1945-7669.
\newblock \doi{10.1257/mic.20220248}.

\bibitem[Magurran(2013)]{magurran2013measuring}
A.~E. Magurran.
\newblock \emph{Measuring biological diversity}.
\newblock John Wiley \& Sons, 2013.

\bibitem[Malerba(1992)]{Malerba1992LearningChange}
F.~Malerba.
\newblock {Learning by Firms and Incremental Technical Change}.
\newblock \emph{The Economic Journal}, 102\penalty0 (413):\penalty0 845--859, 1992.

\bibitem[Matsuyama and Ushchev(2022)]{Matsuyama2022DestabilizingInnovation}
K.~Matsuyama and P.~Ushchev.
\newblock {Destabilizing effects of market size in the dynamics of innovation}.
\newblock \emph{Journal of Economic Theory}, 200, 2022.

\bibitem[McKelvey and Palfrey(1995)]{mckelvey1995quantal}
R.~D. McKelvey and T.~R. Palfrey.
\newblock Quantal response equilibria for normal form games.
\newblock \emph{Games and Economic Behavior}, 10\penalty0 (1):\penalty0 6--38, 1995.

\bibitem[Menger(1871)]{MengerCarl1871PrinciplesEconomics}
C.~Menger.
\newblock \emph{Principles of economics}.
\newblock Ludwig von Mises Institute, 1871.

\bibitem[Mertikopoulos and Viossat(2022)]{MV22}
P.~Mertikopoulos and Y.~Viossat.
\newblock Survival of dominated strategies under imitation dynamics.
\newblock \emph{Journal of Dynamics and Games}, 9\penalty0 (4, William H. Sandholm memorial issue):\penalty0 499--528, October 2022.

\bibitem[Nehring and Puppe(2002)]{Nehring2002ADiversity}
K.~Nehring and C.~Puppe.
\newblock {A theory of diversity}.
\newblock \emph{Econometrica}, 70\penalty0 (3), 2002.

\bibitem[Nehring and Puppe(2004)]{Nehring2004ModellingProduction}
K.~Nehring and C.~Puppe.
\newblock {Modelling cost complementarities in terms of joint production}.
\newblock \emph{Journal of Economic Theory}, 118\penalty0 (2):\penalty0 252--264, 2004.

\bibitem[Nelson and Winter(1977)]{nelson1977search}
R.~R. Nelson and S.~G. Winter.
\newblock In search of useful theory of innovation.
\newblock \emph{Research policy}, 6\penalty0 (1):\penalty0 36--76, 1977.

\bibitem[Parenti et~al.(2017)Parenti, Sidorov, Thisse, and Zhelobodko]{parenti2017cournot}
M.~Parenti, A.~V. Sidorov, J.-F. Thisse, and E.~V. Zhelobodko.
\newblock Cournot, {B}ertrand or {C}hamberlin: toward a reconciliation.
\newblock \emph{International Journal of Economic Theory}, 13\penalty0 (1):\penalty0 29--45, 2017.

\bibitem[Pepall et~al.(2014)Pepall, Richards, and Norman]{pepall2014industrial}
L.~Pepall, D.~Richards, and G.~Norman.
\newblock \emph{Industrial organization: Contemporary theory and empirical applications}.
\newblock John Wiley \& Sons, 2014.

\bibitem[Pliś(2020)]{plis2020}
E.~Pliś.
\newblock Diversity and innovation in economic evolution.
\newblock \emph{Central European Journal of Economic Modelling and Econometrics}, 12\penalty0 (4):\penalty0 347--367, 2020.

\bibitem[Sandholm(2010{\natexlab{a}})]{San10}
W.~H. Sandholm.
\newblock \emph{Population Games and Evolutionary Dynamics}.
\newblock MIT Press, Cambridge, MA, 2010{\natexlab{a}}.

\bibitem[Sandholm(2010{\natexlab{b}})]{sandholm2010local}
W.~H. Sandholm.
\newblock Local stability under evolutionary game dynamics.
\newblock \emph{Theoretical Economics}, 5\penalty0 (1):\penalty0 27--50, 2010{\natexlab{b}}.

\bibitem[Sato and Crutchfield(2003)]{sato2003coupled}
Y.~Sato and J.~P. Crutchfield.
\newblock Coupled replicator equations for the dynamics of learning in multiagent systems.
\newblock \emph{Physical Review E}, 67\penalty0 (1):\penalty0 015206, 2003.

\bibitem[Saviotti and Pyka(2004)]{Saviotti2004EconomicSectors}
P.~P. Saviotti and A.~Pyka.
\newblock {Economic development by the creation of new sectors}.
\newblock \emph{Journal of Evolutionary Economics}, 14\penalty0 (1):\penalty0 1--35, 2004.

\bibitem[Schumpeter(1934)]{Schumpeter2017TheoryCycle}
J.~A. Schumpeter.
\newblock \emph{{Theory of economic development: An inquiry into profits, capital, credit, interest, and the business cycle}}.
\newblock New Brunswick, New Jersey: Transaction Books, 1934.

\bibitem[Schumpeter(1942)]{Schumpeter2017CapitalismDemocracy}
J.~A. Schumpeter.
\newblock \emph{Capitalism, Socialism and Democracy}.
\newblock Harper \& Brothers, 1942.

\bibitem[Sher(2018)]{Sher2018EvaluatingFreedom}
I.~Sher.
\newblock Evaluating allocations of freedom.
\newblock \emph{The Economic Journal}, 128\penalty0 (612):\penalty0 F65--F94, 2018.

\bibitem[Silverberg et~al.(1988)Silverberg, Dosi, and Orsenigo]{Silverberg1988InnovationModel}
G.~Silverberg, G.~Dosi, and L.~Orsenigo.
\newblock Innovation, diversity and diffusion: A self-organisation model.
\newblock \emph{The Economic Journal}, 98\penalty0 (393):\penalty0 1032--1054, 1988.

\bibitem[Smith(1984)]{Smi84}
M.~J. Smith.
\newblock The stability of a dynamic model of traffic assignment - an application of a method of {Lyapunov}.
\newblock \emph{Transportation Science}, 18:\penalty0 245--252, 1984.

\bibitem[Stiglitz(2017)]{stiglitz2017monopolistic}
J.~E. Stiglitz.
\newblock Monopolistic competition, the {D}ixit--{S}tiglitz model, and economic analysis.
\newblock \emph{Research in Economics}, 71\penalty0 (4):\penalty0 798--802, 2017.

\bibitem[Stirling(2007)]{Stirling2007ASociety}
A.~Stirling.
\newblock A general framework for analysing diversity in science, technology and society.
\newblock \emph{Journal of the Royal Society Interface}, 4\penalty0 (15):\penalty0 707--719, 2007.

\bibitem[Taylor and Jonker(1978)]{taylorjonker78}
P.~D. Taylor and L.~Jonker.
\newblock Evolutionarily stable strategies and game dynamics.
\newblock \emph{Mathematical Biosciences}, 40:\penalty0 145--156, 1978.

\bibitem[Webb(2007)]{webb2007game}
J.~N. Webb.
\newblock \emph{Game theory: decisions, interaction and evolution}.
\newblock Springer Science \& Business Media, 2007.

\bibitem[Weitzman(1992)]{weitzman1992diversity}
M.~L. Weitzman.
\newblock On diversity.
\newblock \emph{The Quarterly Journal of Economics}, 107\penalty0 (2):\penalty0 363--405, 1992.

\bibitem[Weitzman(1998)]{weitzman1998noah}
M.~L. Weitzman.
\newblock The {N}oah's ark problem.
\newblock \emph{Econometrica}, pages 1279--1298, 1998.

\end{thebibliography}

\section*{Appendix}

\begin{proof}[Proof of Theorem \ref{thm:ESS2strategies}]
First, let us consider Nash equilibria of the game \eqref{eq:game-gen}. If $(\gaina-\gainc)(\gaind-\gainb)<0$, that is in cases (1) and (2), one of the pure strategies is strictly dominated by the other, which results in the dominance-solvable game. Thus, there exists a unique (strict) Nash equilibrium -- either $(0,1)$ (when $\gaina<\gainc$ and $\gainb<\gaind$) or $(1,0)$ (when $\gaina>\gainc$ and $\gainb>\gaind$). This determines a unique (monomorphic) evolutionarily stable state. 
On the other hand, once $(a-c)(d-b)>0$, the picture is slightly more complicated. In such a case, there exists an interior (mixed) Nash equilibrium of the form  
\begin{equation}
\label{eq:eqN}
(\eq,1-\eq)=\left(\frac{d-b}{a-c+d-b},\frac{a-c}{a-c+d-b}\right),
\end{equation}
In particular, when $a>c$ and $d>b$, the game \eqref{eq:game-gen} is a coordination game with an additional two pure Nash equilibria, that is $(0,1)$ and $(1,0)$. 
Finally, when $a<c$ and $d<b$, game \eqref{eq:game-gen} is an anti-coordination game with a unique Nash equilibrium given by \eqref{eq:eqN}. 

Second, let us check which of the Nash equilibria (and when) is attracting in the game dynamics defined by \eqref{eq:in-out}. 
In case of a population game \eqref{eq:game-gen}, the mean dynamics given by \eqref{eq:in-out}, is boiled down to the equation
\begin{equation}
\label{eq:2mean} \dot{\stratx}=(1-\stratx)\switch_{BA}(\stratx)-\stratx \switch_{AB}(\stratx),
\end{equation}
where $\stratx=\stratx_{A}$ denotes the rate of the population that use strategy $A$, and $1-\stratx$ is the population rate of $B$-strategists. Although payoff functions and conditional switch rates  depend on $\strat=(\stratx,1-\stratx)$, to simplify the notation, in the rest of the article we will treat all functions as dependent only on population rate of $\stratx $. 
Since $\switch_{AB}(\stratx)=0$ in case of $\pay_B(\stratx)<\pay_A(\stratx)$, and $\switch_{BA}(\stratx)=0$ if $\pay_A(\stratx)<\pay_B(\stratx)$, we get
\[\dot{\stratx}=\begin{cases}
   (1-\stratx)\switch_{BA}(\stratx)& \text{if }\pay_A(\stratx)>\pay_B(\stratx)\\   
      -\stratx\switch_{AB}(\stratx)& \text{if }\pay_A(\stratx)<\pay_B(\stratx)\\
\end{cases}.\]
As a consequence, we have
 \begin{equation}\label{eq:dot-pay1} \dot{\stratx}<0,\;\text{ if }  \pay_A(\stratx)< \pay_B(\stratx),\text{ and }
 \dot{\stratx}>0,\;\text{ if } \pay_A(\stratx)> \pay_B(\stratx).\end{equation}
Thus, in case of $(a-c)(d-b)<0$, from the first part of the proof and \eqref{eq:dot-pay1}, the unique Nash equilibrium is globally attracting. 
On the other hand, if $(a-c)(d-b)>0$, then $\pay_A(\stratx)= \pay_B(\stratx)$ only when $\stratx=\eq$ with $\eq$ given by \eqref{eq:eqN}.
Then $a>c$ and $d>b$ guarantee that  $\pay_A(\stratx)< \pay_B(\stratx)$ for $\stratx<\eq$, and $\pay_A(\stratx)> \pay_B(\stratx)$ for $\stratx>\eq$. 
Thus, by \eqref{eq:dot-pay1}, all the states of the population, except for $(\eq,1-\eq)$, are attracted by one of the Nash equilibria: $(0,1)$ or $(1,0)$.
Eventually,  when $a<c$ and $d<b$, once again by \eqref{eq:dot-pay1}, the interior fixed point $\eq$ attracts all points from the open interval $(0,1)$. 
So, for every interior initialization $\stratx(0)\in (0,1)$ the dynamics \eqref{eq:2mean} converges to $\eq$, and the \bv interior \bc population state $(\eq,1-\eq)$ attracts all other states except for monomorphic states $(0,1)$ and $(1,0)$.

Finally, to show that the attracting fixed points of dynamics \eqref{eq:in-out} determine evolutionary stable states, we use Theorem 9.8 from \citep{webb2007game}, which states that any asymptotically stable point of replicator dynamics is also an ESS.  By \eqref{eq:dot-pay1}, the sign of the derivative in equation \eqref{eq:in-out} agrees with the sign of the derivative in the replicator dynamics (as we mentioned before, this equation can be obtained from \eqref{eq:in-out} via pairwise proportional imitation \eqref{eq:PPI}):
\[\dot{\stratx}=\stratx (1-\stratx)(\pay_A(\stratx)-\pay_B(\stratx)),\]
thus, the methods used in the proof of Theorem 9.8 in \citep{webb2007game} work out in our (more general) case. 
This completes the proof.
 
\end{proof}

\begin{proof}[Proof of Proposition \ref{prop:s1}]
Define a function \begin{equation}\label{function}    
\Gamma\colon \domain\ni(\aprof,\bprof)\mapsto \Gamma(\aprof,\bprof)\in(0,1), 
\end{equation} with \begin{equation}\label{domain}
\domain=\left\{
(\aprof,\bprof)\in \mathbb R^2\colon\frac{\aprof}{3}-\frac{\xi}{3}< \bprof<\aprof-\frac{\xi}{2};\;\; \aprof >0
\right\},
\end{equation}
and given by 
\begin{equation}
\Gamma(\aprof,\bprof):=\frac{2(\aprof-\bprof)-\xi}{\aprof+\bprof}.
\end{equation}
 The form of $\domain$ results from the condition \eqref{sch} that restricts our analysis to the population games of the form of the anti-coordination games, in which there exists the (unique) \bv interior \bc and evolutionarily stable form of population.
 
Partial derivatives of $\Gamma$ are
\begin{equation}\label{prox1}
\frac{\partial \Gamma}{\partial \aprof}(\aprof,\bprof)=\frac{4\bprof+\xi}{(\aprof+\bprof)^2} \;\;\;\text{ and }\;\;\;\;\;
\frac{\partial \Gamma}{\partial \bprof}(\aprof,\bprof)=\frac{-4\aprof+\xi}{(\aprof+\bprof)^2}.
\end{equation}
Since $\frac{\partial \Gamma}{\partial \aprof}(\aprof,\bprof)>0$, and $\frac{\partial \Gamma}{\partial \bprof}(\aprof,\bprof)<0$ when $(\aprof,\bprof)\in \domain$, the function $ \Gamma (\cdot,\bprof)\colon \aprof \mapsto \Gamma(\aprof,\bprof)$ increases and $ \Gamma (\aprof,\cdot)\colon \bprof \mapsto \Gamma(\aprof,\bprof)$ decreases with $(\aprof,\bprof)\in \domain$. This completes the proof of this proposition.
\end{proof}

\end{document}